\documentclass[11pt]{article}       

\usepackage[a4paper,left=35mm,right=35mm,top=20mm,bottom=25mm,includefoot]{geometry}

\usepackage{amsmath}
\usepackage{amsthm}
\newtheorem{theorem}{Theorem}
\newtheorem{lemma}[theorem]{Lemma}
\newtheorem{fact}[theorem]{Fact}
\newtheorem{corollary}[theorem]{Corollary}
\theoremstyle{definition}
\newtheorem{algorithm}[theorem]{Algorithm}

\usepackage[tt=false]{libertine}
\usepackage{array}
\newcolumntype{C}{>{$}c<{$}}
\newcolumntype{R}{>{$}r<{$}}
\usepackage{bbold}
\usepackage{bigstrut}
\usepackage{booktabs}
\usepackage{csquotes}
\usepackage{xcolor}
\colorlet{lightgray}{black!10!white}

\usepackage{colortbl}
\usepackage{graphicx}
\usepackage{tikz}
\usepackage{xspace}
\usepackage{todonotes}
\usepackage{multirow}

\newcommand{\card}[1]{|#1|}
\newcommand{\from}{:}
\newcommand{\etc}{etc.\xspace}
\newcommand{\ie}{i.\,e.\xspace}
\newcommand{\pcm}{p_c}
\newcommand{\pmax}{p_{\text{max}}}

\newcommand{\Per}{P}
\newcommand{\tc}{t_c}
\newcommand{\msca}{{\normalfont\textsc{MS-CA}\xspace}}
\newcommand{\msfssp}{{\normalfont\textsc{MS-FSSP}\xspace}}
\newcommand{\N}{\mathbb{N}}
\DeclareMathOperator{\lcm}{\mathrm{lcm}}
\DeclareMathOperator{\pfx}{\mathrm{pfx}}
\DeclareMathOperator{\sfx}{\mathrm{sfx}}
\DeclareMathOperator{\supp}{\mathrm{supp}}
\newcommand{\sB}{\texttt{\#}}
\newcommand{\sG}{\texttt{G}}
\newcommand{\sS}{\texttt{\char95}}
\newcommand{\sF}{\texttt{F}}

\newcommand{\sL}{\texttt{<}}
\newcommand{\sR}{\texttt{>}}
\newcommand{\x}{\times}
\newcommand{\Z}{\mathbb{Z}}

\usepackage{xparse}
\NewDocumentCommand\statebox{m}{
  \hbox to 4ex{%
    \hss
    \vrule height 3.8ex depth -0.3ex%
    \vbox to 3.8ex{%
      \hrule width 3.8ex height 0.1ex %
      \vss
      \hbox to 3.8ex{\hss #1\hss}%
      \vss
      \hrule width 3.8ex height 0.1ex %
    }%
    \vrule height 3.8ex depth -0.3ex%
    \hss
  }
}
\NewDocumentCommand\cell{m}{
  \statebox{#1}%
}
\NewDocumentCommand\cR{t+}{
  \IfBooleanTF#1%
  {\cell{\sR\hbox to 0pt{\hspace*{1pt}\raisebox{0.5pt}{$\cdot$}}}}
  {\cell{\sR}}%
}
\NewDocumentCommand{\cL}{t+}{
  \IfBooleanTF{#1}%
  {\cell{\hbox to 0pt{\hspace*{1pt}\raisebox{0.5pt}{$\cdot$}}\sL}}%
  {\cell{\sL}}%
}
\NewDocumentCommand{\cG}{o}{
  \cell{\texttt{\sG}}%
}
\NewDocumentCommand{\cH}{o}{
  \cell{\texttt{\#}}%
}
\NewDocumentCommand\cE{ }{
  \cell{\texttt{ }}%
}
\newcommand{\ac}{\tikz\fill[red](0,0)--(3ex,0)--(1.5ex,-1mm)--cycle;}
\newcommand{\nc}{\tikz\fill[white](0,0)--(3ex,0)--(1.5ex,-1mm)--cycle;}
\newcommand{\cu}{\multirow{2}{4em}[3ex]{ common\\ \ update}}

\begin{document}

\title{A faster algorithm for the FSSP in one-dimensional CA with
  multiple speeds}

\author{Thomas Worsch\\
  Karlsruhe Institute of Technology,
  Department of Informatics \\
  \texttt{worsch at kit.edu}
}

\maketitle

\begin{abstract}
  In cellular automata with multiple speeds for each cell $i$ there is
  a positive integer $p_i$ such that this cell updates its state still
  periodically but only at times which are a multiple of $p_i$.
  Additionally there is a finite upper bound on all $p_i$.

  Manzoni and Umeo have described an algorithm for these
  (one-dimen\-sio\-nal) cellular automata which solves the Firing Squad
  Synchronization Problem.
  This algorithm needs linear time (in the number of cells to be
  synchronized) but for many problem instances it is slower than the
  optimum time by some positive constant factor.
  In the present paper we derive lower bounds on possible
  synchronization times and describe an algorithm which is never
  slower and in some cases faster than the one by Manzoni and Umeo and
  which is close to a lower bound (up to a constant summand) in more
  cases.
\end{abstract}

\section{Introduction}
\label{s:introduction}

The \emph{Firing Squad Synchronization Problem (FSSP)} has a
relatively long history in the field of cellular automata.
The formulation of the problem dates back to the late fifties and
first solutions were published in the early sixties.
A general overview of different variants of the problem and solutions
with many references can be found in \cite{Umeo_2009_FSSP_ic}.
Readers interested in more recent developments concerning several
specialized problems and questions are referred to the survey
\cite{Umeo_2020_HSCA_ar}.

In recent years asynchronous CA have received a lot of attention.
In a \enquote{really} asynchronous setting (when nothing can be
assumed about the relation between updates of different cells) it is
of course impossible to achieve synchronization.
As a middle ground the FSSP has been considered in what Manzoni and
Umeo \cite{Manzoni_2014_FSSP_ar} have called \emph{CA with multiple
  speeds}, abbreviated in the following as \msca.
In these CA different cells may update their states at different
times.
But there is still enough regularity so that the problem setting of
the FSSP makes sense:
As in standard CA there is a global clock.
For each cell $i$ there is a positive integer $p_i$ such that this
cell only updates its state at times $t$ which are a multiple of
$p_i$.
We will call $p_i$ the \emph{period} of cell $i$.
Additionally there is a finite upper bound on all $p_i$, so that it
can be assumed that each cell has stored $p_i$ as part of its
state.
This also means that there are always times (namely the multiples of
the least common multiple of all periods) when all cells update their
states simultaneously.

The rest of this paper is organized as followed:
In Section~\ref{s:basics} we fix some notation, review the basics of
standard CA in general and of the FSSP for them.
In Section~\ref{s:msca} cellular automata with multiple speeds (\msca)
and the corresponding FSSP will be introduced.
Since most algorithms for the FSSP make heavy use of signals, we have
a closer look at what can happen with them in \msca.
In Section~\ref{s:lower-bound} some lower bounds for the
synchronization times will be derived.
Finally, in Section~\ref{s:details} an algorithm for the FSSP in
\msca{} will be described in detail.

\section{Basics}
\label{s:basics}

$\Z$ denotes the set of integers, $\N_+$ the set of positive integers
and $\N_0=\N_+\cup\{0\}$.
For $k\in \N_0$ and $M\subset\Z$ we define
$k\cdot M=\{km \mid m\in M\}$.
For $k\in\N_+$ let $\N_k=\{i\in\N_+ \mid 1\leq i\leq n\}$.

The greatest common divisor of a set $M$ of numbers is abbreviated as
$\gcd M$ and the least common multiple as $\lcm M$.

We write $B^A$ for the set of all functions $f\from A\to B$.
The cardinality of a set $A$ is denoted $\card{A}$.

For a finite alphabet $A$ and $k\in\N_0$ we write $A^k$ for the set of
all words over $A$ having length $k$, $A^{\leq k}$ for
$\bigcup_{i=0}^k A^i$, and $A^*=\bigcup_{i=0}^{\infty} A^i$.
%
For a word $w\in A^*$ and some $k\in\N_0$ the longest prefix of $w$
which has length at most $k$ is denoted as $\pfx_k(w)$, \ie
$\pfx_k(w)$ is the prefix of length $k$ of $w$ or the whole word $w$
if it is shorter than $k$.
Analogously $\sfx_k(w)$ is used for suffixes of $w$.

Usually cellular automata are specified by a finite set $S$ of states,
a neighborhood $N$, and local transition function $f: S^N\to S$.
In the present paper we will only consider one-dimensional CA with
Moore neighborhood $N=\{-1,0,1\}$ with radius $1$.

Therefore a (global) \emph{configuration} of a CA is a function
$c\from \Z\to S$, \ie $c\in S^{\Z}$.
Given a configuration $c$ and some cell $i\in\Z$ the so-called
\emph{local configuration observed by $i$ in $c$} is the mapping
$g\from N \to S: n \mapsto c(i+n)$ and denoted by $c_{i+N}$.
In the standard definition of CA the local transition function $f$
induces a \emph{global transition function} $F\from S^{\Z} \to S^{\Z}$
describing one step of the CA from $c$ to $F(c)$ by requiring that
$F(c)(i) = f(c_{i+N})$ holds for each $i\in\Z$.

By contrast in \msca{} it is not possible to speak about \emph{the}
successor configuration.
The relevant definitions will be given and discussed in the next
section.

Before, we quickly recap the \emph{Firing Squad Synchronization
  Problem} (FSSP).
A CA solving the FSSP has to have a set of states
$S \supseteq \{\sB, \sG, \sS, \sF\}$.
For $n\in\N_+$ the \emph{problem instance of size $n$} is the
initial configuration $I_n=c$ where
\begin{alignat*}{4}
  \begin{alignedat}{1}
    c(i) = \sB && \text{ \ if } i\leq 0
  \end{alignedat}
  &\hspace*{1.5em}&
  c(1) = \sG
  &\hspace*{1.5em}&
  \begin{alignedat}{1}
    c(i) = \sS && \text{ \ if } 1< i \leq n 
  \end{alignedat}
  &\hspace*{1.5em}&
  \begin{alignedat}{1}
    c(i) = \sB && \text{ \ if } i>n 
  \end{alignedat}
\end{alignat*}
Cells outside the segment between $1$ and $n$ are in a state $\sB$
which is supposed to be fixed by all transitions.
State $\sG$ is the initial state for the so-called \emph{general},
state $\sS$ is the initial state for all other cells, and state $\sF$
indicates the cells have been synchronized.

The goal is to find a local transition function which makes the CA
transit from each $I_n$ to configuration $F_n$ in which all cells
$1,\dots,n$ are in state $\sF$ (the other cells still all in state
$\sB$) and no cell ever was in state $\sF$ before.
In addition the local transition function has to satisfy
$f(\sS,\sS,\sS)=\sS$ and $f(\sS,\sS,\sB)=\sS$, which prohibits the
trivial \enquote{solution} to have all cells enter state $\sF$ in the
first step and implies that \enquote{activities} have to start at the
\sG-cell and spread to other cells from there.

Within the framework of synchronization let's call the set
$\supp(c)=\{i\in \Z \mid c(i)\not=\sB\}$ the \emph{support} of
configuration $c$.
As a consequence of all these requirements during a computation
starting with some problem instance $I_n$ all subsequent
configurations have the same support $\N_n$.

It is well known that there are CA which achieve synchronization in
time $2n-2$ (for $n\geq 2$) and that no CA can be faster, not even for
a single problem instance.

\section{CA with multiple speeds}
\label{s:msca}

\subsection{Definition of \msca}
\label{subs:msca}

A cellular automaton with multiple speeds (\msca) is a specialization
of standard CA.
Its specifiation requires a \emph{finite} set $\Per\subseteq \N_+$ of
so-called possible \emph{periods} (in \cite{Manzoni_2014_FSSP_ar} they
are called lengths of update cycles).
Before a computation starts a period has to be assigned to each cell
which remains fixed throughout the computation.
Requiring $\Per$ to be finite is meaningful for several reasons:
\begin{itemize}
\item It can be shown \cite[Prop.~3.1]{Manzoni_2014_FSSP_ar} that
  otherwise it is impossible to solve the FSSP for \msca{} with one
  fixed set of states.
\item We want that each cell can make its computation depend on its
  own period and those of its neighbors to the left and right, but of
  course the analogue of the local transition function should still
  have a finite description.
  To this end we want to be able to assume that each has its own
  period stored in its state.
\end{itemize}
For the rest of the paper we assume that the set of states is always
of the form $S=\Per\x S'$, and that the transition function
\emph{never changes} the first component.
We will denote the period of a cell $i$ as $p_i$.
For $s=(p,s')\in S=L\x S'$ we write $\pi_{p}$ and $\pi_s$ for the
projections on the first and second component and analogously for
global configurations.
For a global configuration $c\in S^{\Z}$ we write
$P(c)=\{p_i \mid i\in\supp(c)\}$ (or simply $P$ if $c$ is clear from
the context) for the set of numbers that are periods of cells in the
support of $c$; cells in state $\sB$ can be ignored because they
don't change their state by definition.

For \msca{} it is not possible to speak about \emph{the} successor
configuration.
Instead it is necessary to know how many time steps have already
happened since the CA started.
Borrowing some notation from asynchronous CA, for any subset
$A\subseteq\Z$ of cells and any $c\in S^{\Z}$ denote by $F_A(c)$ the
configuration reached from $c$ if exactly the cells in $A$ update
their states according to $f$ and all other cells do not change their
state:
\begin{equation*}
  \forall i\in\Z : F_A(c)(i)=
  \begin{cases}
    f(c_{i+N}) & \text{if } i\in A \\
    c(i) & \text{if } i\notin A
  \end{cases}
\end{equation*}
\noindent
(Thus, the global transition function of a standard CA is $F_{\Z}$.)
Given some \msca{} $C$ and some time $t$ denote by
$A_t = \{ i\in\Z \mid 0= t\bmod p_i \}$ the set of so-called
\emph{active} cells at time $t$.
Then, for each initial configuration $c$ the computation resulting
from it is the sequence $(c_0,c_1,c_2,\dots)$ where $c_0 = c$ and
$c_{t+1} = F_{A_t}(c_t)$ for each $t\in\N_0$.
In particular this means, that at time $t=0$ \emph{all} cells will
update their states according to $f$.
More generally this is true for all $t\in\pcm\cdot\N_0$, where $\pcm$
is the least common multiple of all $p\in\Per$, \ie all $t$ that are a
multiple of all elements in $\Per$.
We will speak of a \emph{common update} when $t=\pcm$.

The observations collected in the following lemma are very simple and
don't need an explicit proof.

\begin{lemma}
  Let $g=\gcd P$ be the greatest common divisor of all $p\in P$ and
  let $P'=\{ p/g \mid p\in P\}$.
  Let $A_t=\{ i\in\Z \mid 0= t\bmod p_i \}$ as before.
  
  \begin{enumerate}
  \item For each $t\notin g\cdot\N_+$ the set $A_t$ is empty.
  \item For each computation $C=(c_0,c_1,\dots)$ for each $k\in\N_0$
    $c_{kg+1} = c_{kg+2} = \cdots = c_{kg+g}$.
  \item The computation $C'=(c_0,c_g,c_{2g},\dots)$ results when using
    $P'$ instead of $P$ and exactly the same local transition function.
  \item When all cells involved in a computation $C$ have the same
    period $p$, $C$ is simply a $p$ times slower \enquote{copy} of the
    computation in a standard CA.
  \end{enumerate}
\end{lemma}
Therefore the interesting cases are whenever $\card{P}\geq 2$ and
$\gcd P=1$.
We will assume this for the rest of the paper without always
explicitly mentioning it again.

\begin{fact}
  \label{fct:odd-period}
  If $\card{P}\geq 2$ and $\gcd P=1$ then there is at least one
  \emph{odd} number $p$ that can be used as a period.
\end{fact}

\subsection{Signals in \msca}
\label{subs:signals}

Since almost all CA algorithms for the synchronization problem make
extensive use of signals, they are also our first example for some
\msca.
Figure~\ref{fig:signal-msca} shows a sketch of a space-time diagram.
Time is increasing in the downward direction (throughout this paper).
When a cell is active at time $t$ a triangle between the old state
in row $t$ and the new state in row $t+1$ indicates the state
transition.
The numbers in the top row are the periods of the cells.

At this point it is not important to understand how an appropriate
local transition function could be designed to realize the depicted
signal.
But, assuming that this can be done, the example has been chosen
such that the signal is present in a cell $i$ for the first time
exactly $p_i$ steps after it was first present in the left neighbor
$i-1$.

\begin{figure}[htb]
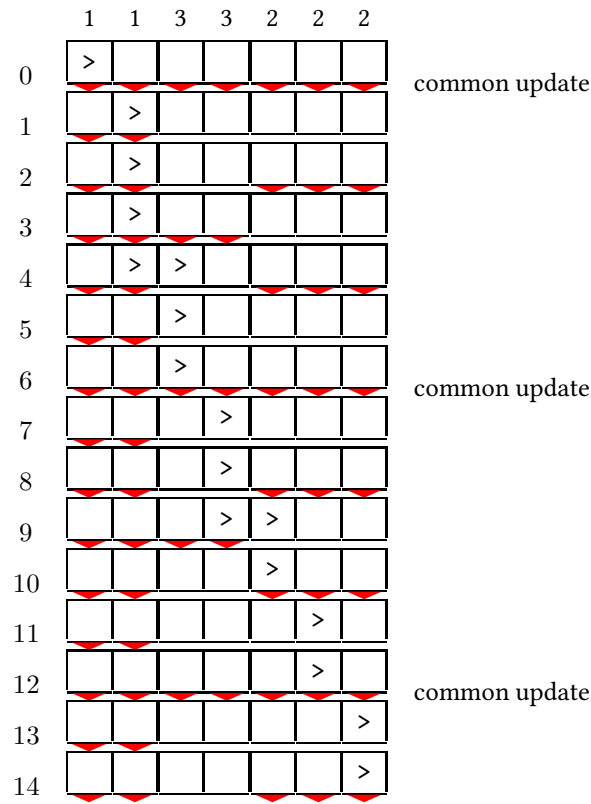

  \centering
  \small
  \setlength{\tabcolsep}{0pt}
  \renewcommand{\arraystretch}{0}
  \begin{tabular}[t]{C@{\quad}*{7}{c}@{\quad}l}
       & 1  & 1 & 3 & 3 & 2 & 2 & 2 \\[2mm]
    0  &\cR &\cE &\cE &\cE &\cE &\cE &\cE & \\
       &\ac &\ac &\ac &\ac &\ac &\ac &\ac & \multirow{2}{*}[2ex]{common update} \\
    1  &\cE &\cR &\cE &\cE &\cE &\cE &\cE & \\
       &\ac &\ac &    &    &    &    &    & \\
    2  &\cE &\cR &\cE &\cE &\cE &\cE &\cE & \\
       &\ac &\ac &    &    &\ac &\ac &\ac & \\
    3  &\cE &\cR &\cE &\cE &\cE &\cE &\cE & \\
       &\ac &\ac &\ac &\ac &    &    &    & \\
    4  &\cE &\cR &\cR &\cE &\cE &\cE &\cE & \\
       &\ac &\ac &    &    &\ac &\ac &\ac & \\
    5  &\cE &\cE &\cR &\cE &\cE &\cE &\cE & \\
       &\ac &\ac &    &    &    &    &    & \\
    6  &\cE &\cE &\cR &\cE &\cE &\cE &\cE & \\
       &\ac &\ac &\ac &\ac &\ac &\ac &\ac & \multirow{2}{*}[2ex]{common update} \\
    7  &\cE &\cE &\cE &\cR &\cE &\cE &\cE & \\
       &\ac &\ac &    &    &    &    &    & \\
    8  &\cE &\cE &\cE &\cR &\cE &\cE &\cE & \\
       &\ac &\ac &    &    &\ac &\ac &\ac & \\
    9  &\cE &\cE &\cE &\cR &\cR &\cE &\cE & \\
       &\ac &\ac &\ac &\ac &    &    &    & \\
    10 &\cE &\cE &\cE &\cE &\cR &\cE &\cE & \\
       &\ac &\ac &    &    &\ac &\ac &\ac & \\
    11 &\cE &\cE &\cE &\cE &\cE &\cR &\cE & \\
       &\ac &\ac &    &    &    &    &    & \\
    12 &\cE &\cE &\cE &\cE &\cE &\cR &\cE & \\
       &\ac &\ac &\ac &\ac &\ac &\ac &\ac & \multirow{2}{*}[2ex]{common update} \\
    13 &\cE &\cE &\cE &\cE &\cE &\cE &\cR & \\
       &\ac &\ac &    &    &    &    &    & \\
    14 &\cE &\cE &\cE &\cE &\cE &\cE &\cR & \\
       &\ac &\ac &    &    &\ac &\ac &\ac & \\
  \end{tabular}
  \caption{Sketch of how a basic signal could move right in a
    \msca as fast as possible.
    The numbers in the top row are the periods of the cells.
    Time is going down.
    Triangles indicate active state transitions, \ie when they are
    missing the state has to stay the same. }
  \label{fig:signal-msca}
\end{figure}

One possibility to construct such computations is the following.
Putting aside an appropriate number of cells at either end, the
configuration consists of blocks of cells.
In each block all cells have some common period $p$ and there are
$k=\pcm/p$ cells in the block.
For example, in Figure~\ref{fig:signal-msca} $\pcm=6$ (at least one
can assume that, since only periods $1$, $2$, and $3$ are used) and
there are $2$ cells with period $3$ and $3$ cells with period $2$.
Let's number the cells in such a block from $1$ to $k$.

Assume that a signal should move to the right as fast as posible.
For each such block the following holds:
If the signal appears in the left neighbor of such a block for the
first time after a common update at some time $t$, then it can only
enter cell $1$ of the block at time $t+p$, cell $2$ of the block at
time $t+2p$ \etc, and hence the last cell $k$ of the block at time
$t+kp=t+\pcm$.
That cell is the left neighbor of the next block.
Hence by induction the same is true for every block and the passage of
the signal through each block will take $\pcm$ steps.
This happens to be the sum of all periods of cells in one block.

Unfortunately there are also cases in which signals are \emph{not}
delayed by the increase of the periods of some cells.
Figure~\ref{fig:signal-not-slow} shows a situation where periods $1$
and $2$ are assigned to subsequent cells alternatingly.
As can be seen, a signal can move from each cell to next one in every
step, including a change of direction at the right border.

\begin{figure}[htb]
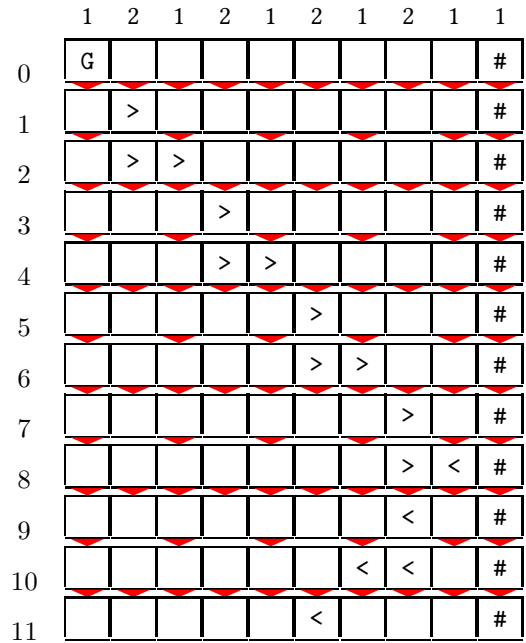

  \centering
  \small
  \setlength{\tabcolsep}{0pt}
  \renewcommand{\arraystretch}{0}
  \begin{tabular}{C@{\quad}*{10}{c}@{\quad}l}
      & 1  & 2  & 1  & 2  & 1  & 2  & 1  & 2  & 1  & 1   \\[2mm]
    0 &\cG &\cE &\cE &\cE &\cE &\cE &\cE &\cE &\cE &\cH  \\
      &\ac &\ac &\ac &\ac &\ac &\ac &\ac &\ac &\ac &\ac  \\
    1 &\cE &\cR &\cE &\cE &\cE &\cE &\cE &\cE &\cE &\cH  \\
      &\ac &    &\ac &    &\ac &    &\ac &    &\ac &\ac  \\
    2 &\cE &\cR &\cR &\cE &\cE &\cE &\cE &\cE &\cE &\cH  \\
      &\ac &\ac &\ac &\ac &\ac &\ac &\ac &\ac &\ac &\ac  \\
    3 &\cE &\cE &\cE &\cR &\cE &\cE &\cE &\cE &\cE &\cH  \\
      &\ac &    &\ac &    &\ac &    &\ac &    &\ac &\ac  \\
    4 &\cE &\cE &\cE &\cR &\cR &\cE &\cE &\cE &\cE &\cH  \\
      &\ac &\ac &\ac &\ac &\ac &\ac &\ac &\ac &\ac &\ac  \\
    5 &\cE &\cE &\cE &\cE &\cE &\cR &\cE &\cE &\cE &\cH  \\
      &\ac &    &\ac &    &\ac &    &\ac &    &\ac &\ac  \\
    6 &\cE &\cE &\cE &\cE &\cE &\cR &\cR &\cE &\cE &\cH  \\
      &\ac &\ac &\ac &\ac &\ac &\ac &\ac &\ac &\ac &\ac  \\
    7 &\cE &\cE &\cE &\cE &\cE &\cE &\cE &\cR &\cE &\cH  \\
      &\ac &    &\ac &    &\ac &    &\ac &    &\ac &\ac  \\
    8 &\cE &\cE &\cE &\cE &\cE &\cE &\cE &\cR &\cL &\cH  \\
      &\ac &\ac &\ac &\ac &\ac &\ac &\ac &\ac &\ac &\ac  \\
    9 &\cE &\cE &\cE &\cE &\cE &\cE &\cE &\cL &\cE &\cH  \\
      &\ac &    &\ac &    &\ac &    &\ac &    &\ac &\ac  \\
   10 &\cE &\cE &\cE &\cE &\cE &\cE &\cL &\cL &\cE &\cH  \\
      &\ac &\ac &\ac &\ac &\ac &\ac &\ac &\ac &\ac &\ac  \\
   11 &\cE &\cE &\cE &\cE &\cE &\cL &\cE &\cE &\cE &\cH  \\
  \end{tabular}
  \caption{Basic signal first moving right and bouncing back at the
    border with speed $1$ although half of the cells has only period
    $2$.}
  \label{fig:signal-not-slow}
\end{figure}

\subsection{The FSSP in \msca}
\label{subs:fssp-msca}

In the standard setting for each $n$ there is exactly one problem
instance of the FSSP of length $n$.

In \msca{} we will assume that the set of states is always of the form
$S=\Per\x S'$.
We will call a configuration $c\in S^{\Z}$ a \emph{problem instance
  for the \msfssp} if two conditions are satisfied:
\begin{itemize}
\item $\pi_s(c)$ is a problem instance for the FSSP in standard CA.
\item The period of all border cells is the same as that of \sG-cell
  $1$.
\end{itemize}
By definition border cells never change their state, no matter what
their period is.
The second condition just makes sure that formally a period is
assigned even to border cells, but this does not change the set of
periods that are present in the cells of the support that do the real
work.

Now for each size $n$ there are $\card{\Per}^n$ problem instances of
the \msfssp.

It should be clear that the minimum synchronization time will at least
in some cases depend on the periods.
Assume that there are two different $p, q\in \Per$, say $p<q$.
Then, when all cells have $p_i=p$ synchronization can be
achieved more quickly than when all cells have $p_i=q$.
A straightforward transfer of a (time optimal) FSSP algorithm for
standard CA (needing $2n+O(1)$ steps) yields a \msca{} running in
time $(2n-2)p$.
This is faster than any \msca{} with uniform period $q$ can
be which needs $(2n-2)q$ (see Section~\ref{s:lower-bound}).

\section{On lower bounds for the synchronization time on \msca}
\label{s:lower-bound}

In the case of standard CA the argument used for deriving lower bounds
for the synchronization time uses the following observation.
Whenever an algorithm makes the leftmost cell $1$ fire at some time
$t$, it can only be correct if changing the border state $\sB$ in cell
$n+1$ to state $\sS$ (\ie increasing the size of the initial
configuration by $1$) can possibly (and in fact will) have an
influence on the state of cell $1$ at time $t$.
If $t\leq 2n-3$ changing the state at the right end cannot have an
influence on cell $1$.
But then adding $n$ cells in state $\sS$ to the right will still make
cell $1$ enter state $\sF$ at time $t$, while the now rightmost cell
$2n$ will not have had any chance to leave its state.

This argument can of course be transferred to \msca, and it means that
one has to find out the minimum time to send a signal to the rightmost
cell of the support and back to cell $1$.

\begin{theorem}
  \label{thm:lower-bound}
  For every \msca{} solving the \msfssp{} there are constants $a>1$
  and $d$ such that for infinitely many $n\geq 2$ there are at least
  $a^n$ problem instances $c$ of size $n$ such that $C$ needs at
  least
  \begin{equation*}
    \left( 2\sum_{i=1}^{n}  p_i \right) - d
  \end{equation*}
  steps for synchronization of $c$.
\end{theorem}

\begin{proof} 
  The example in Figure~\ref{fig:signal-msca} can be generalized.

  We first define a set $M$ of periods we will make use of.
  According to Fact~\ref{fct:odd-period} the set
  $P'=\{p \in P\mid p \text{ is odd} \}$ is not empty.
  If $\card{P'}\geq 2$ then let $M=P'$ and $q=\max M$.
  If $\card{P'}=1$ then let $M=P'\cup\{q\}$ where
  $q = \max( P\setminus P')$ (choosing the maximum is not important;
  we just want to be concrete).
  Let $m_c=\lcm M$.

  We will use blocks of length $b_p=m_c/p$ of sucessive cells with the
  same period $p$ (as in Figure~\ref{fig:signal-msca}).
  As will be seen it is useful to have at least one odd $b_p$.
  Indeed, $b_q=m_c/q$ is odd (because $q$ is the only possibly even
  number in $M$).

  Since there are at least $2$ different numbers $p\in M$, there are
  also at least $2$ different block lengths in
  $B=\{ b_p \mid p\in M\}$, denoted as $b$ and $b'$.

  Consider now all problem instances similar to
  Figure~\ref{fig:signal-msca}.
  Let $m\in \N_+$
  \begin{itemize}
  \item The periods of the first $2$ cells are the same but otherwise
    arbitrary.
  \item The rest of the cells is partitioned into $2m$ blocks.
    There are $m$ blocks consisting of $b$ cells and the period of all
    cells in that block is $m_c/b$ and there are $m$ blocks consisting
    of $b'$ cells and the period of all cells in that block is
    $m_c/b'$.
  \item At the right end of the instance there is a segment consisting
    of $1+h$ cells of period $q$ where $h=\lfloor b_q/2\rfloor$.
    Since $b_q$ is odd, one has $2h+1=b_q$.
  \end{itemize}

  For each $m$ the total size of the problem instances is
  $2+m(b+b')+1+h$ which is linear in $m$, and there are $\binom{2m}{m}$
  different arrangements of the blocks.
  This number is known to be larger than $4^m/(2m+1)$ (proof by induction).
  Formulated the other way around for these problem sizes $n$ there is
  a number of problem instances which exponential in $m$ and hence
  also in $n$ (for some appropriately chosen base $a>1$).

  It remains to estimate the synchronization time for these problem
  instances.
  As already described in Subsection~\ref{subs:signals} a signal that
  is supposed to first move to the right border as fast as possible
  and then back to cell $1$ will arrive in cell $2$ after the first
  (common) update.
  From that time on for each block it will take exactly $m_c$ steps to
  \enquote{traverse} each block which is also the sum of all periods
  of the cells in the block.

  For the passage through the last $h+1$ cells forth and back have a
  look at Figure~\ref{fig:reflection}.
  Each cell is passed twice, once when the signal moves right and once
  when it moves back to the left, each time for $q$ steps.
  The only exception is the rightmost cell, where the signal stays
  only for the duration of $1$ period, \ie $q$ steps.
  Altogether these are $(2h+1) q=m_c$ steps which is exactly the same
  number of steps as for each full block to the left.
  Consequently the position of the signal moving back to the left is
  for the first time in the cell to the right of a full block
  immediatly after a common update.
  
  \begin{figure}[!h]
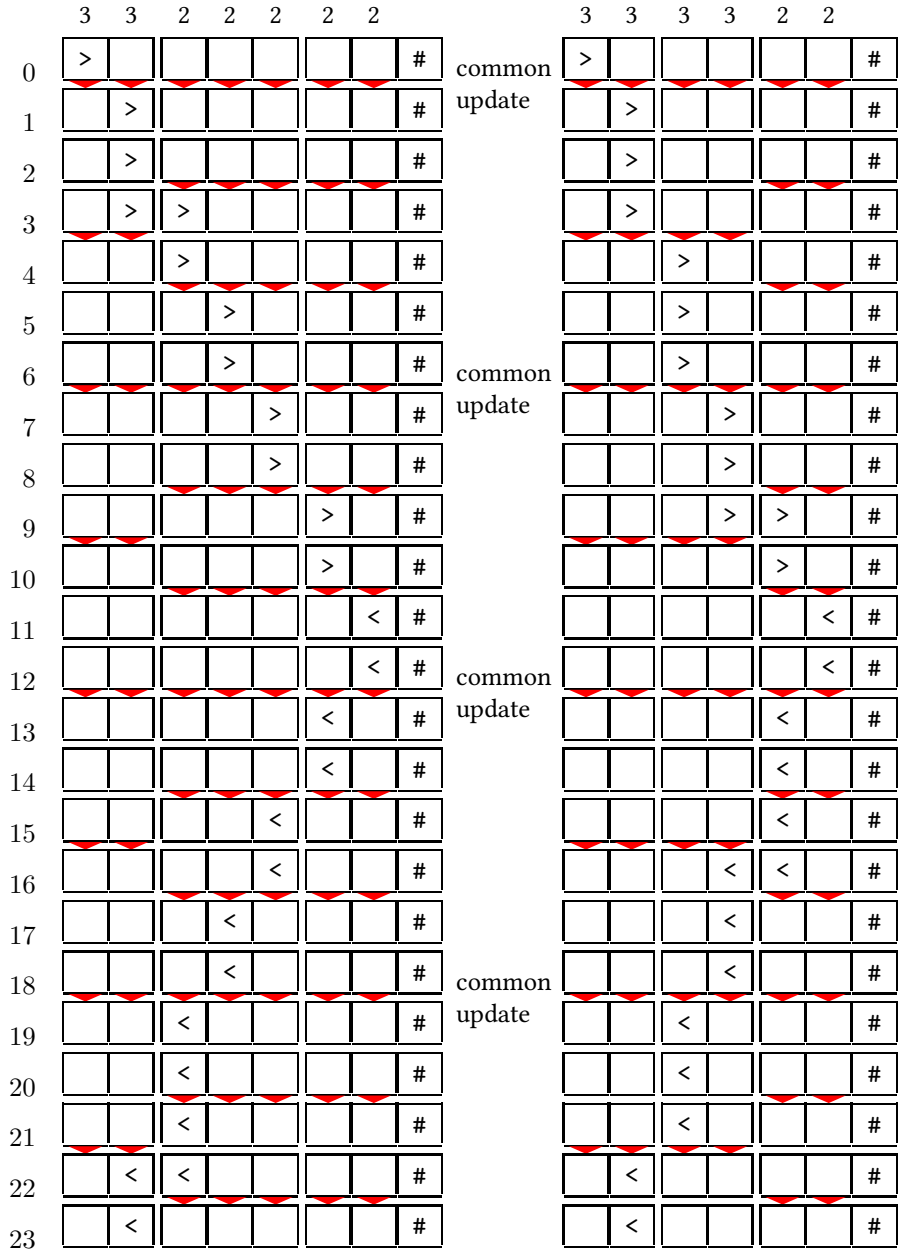

    \centering
    \small
    \setlength{\tabcolsep}{0pt}
    \renewcommand{\arraystretch}{0}
    \begin{tabular}[t]{R@{\quad}cc@{ }ccc@{ }ccc@{ }c@{ }cc@{ }cc@{ }cccc}
       & 3  & 3  & 2  & 2  & 2  & 2  & 2  &    &    & 3  & 3  & 3  & 3  & 2  & 2  &    &  \\[2mm]
    0  &\cR &\cE &\cE &\cE &\cE &\cE &\cE &\cH &    &\cR &\cE &\cE &\cE &\cE &\cE &\cH & \\
       &\ac &\ac &\ac &\ac &\ac &\ac &\ac &    &\cu &\ac &\ac &\ac &\ac &\ac &\ac &    & \\
    1  &\cE &\cR &\cE &\cE &\cE &\cE &\cE &\cH &    &\cE &\cR &\cE &\cE &\cE &\cE &\cH & \\
       &    &    &    &    &    &    &    &    &    &    &    &    &    &    &    &    & \nc \\
    2  &\cE &\cR &\cE &\cE &\cE &\cE &\cE &\cH &    &\cE &\cR &\cE &\cE &\cE &\cE &\cH & \\
       &    &    &\ac &\ac &\ac &\ac &\ac &    &    &    &    &    &    &\ac &\ac &    & \\
    3  &\cE &\cR &\cR &\cE &\cE &\cE &\cE &\cH &    &\cE &\cR &\cE &\cE &\cE &\cE &\cH & \\
       &\ac &\ac &    &    &    &    &    &    &    &\ac &\ac &\ac &\ac &    &    &    & \\
    4  &\cE &\cE &\cR &\cE &\cE &\cE &\cE &\cH &    &\cE &\cE &\cR &\cE &\cE &\cE &\cH & \\
       &    &    &\ac &\ac &\ac &\ac &\ac &    &    &    &    &    &    &\ac &\ac &    & \\
    5  &\cE &\cE &\cE &\cR &\cE &\cE &\cE &\cH &    &\cE &\cE &\cR &\cE &\cE &\cE &\cH & \\
       &    &    &    &    &    &    &    &    &    &    &    &    &    &    &    &    & \nc \\
    6  &\cE &\cE &\cE &\cR &\cE &\cE &\cE &\cH &    &\cE &\cE &\cR &\cE &\cE &\cE &\cH & \\
       &\ac &\ac &\ac &\ac &\ac &\ac &\ac &    &\cu &\ac &\ac &\ac &\ac &\ac &\ac &    & \\
    7  &\cE &\cE &\cE &\cE &\cR &\cE &\cE &\cH &    &\cE &\cE &\cE &\cR &\cE &\cE &\cH & \\
       &    &    &    &    &    &    &    &    &    &    &    &    &    &    &    &    & \nc \\
    8  &\cE &\cE &\cE &\cE &\cR &\cE &\cE &\cH &    &\cE &\cE &\cE &\cR &\cE &\cE &\cH & \\
       &    &    &\ac &\ac &\ac &\ac &\ac &    &    &    &    &    &    &\ac &\ac &    & \\
    9  &\cE &\cE &\cE &\cE &\cE &\cR &\cE &\cH &    &\cE &\cE &\cE &\cR &\cR &\cE &\cH & \\
       &\ac &\ac &    &    &    &    &    &    &    &\ac &\ac &\ac &\ac &    &    &    & \\
   10  &\cE &\cE &\cE &\cE &\cE &\cR &\cE &\cH &    &\cE &\cE &\cE &\cE &\cR &\cE &\cH & \\
       &    &    &\ac &\ac &\ac &\ac &\ac &    &    &    &    &    &    &\ac &\ac &    & \\
   11  &\cE &\cE &\cE &\cE &\cE &\cE &\cL &\cH &    &\cE &\cE &\cE &\cE &\cE &\cL &\cH & \\
       &    &    &    &    &    &    &    &    &    &    &    &    &    &    &    &    & \nc \\
   12  &\cE &\cE &\cE &\cE &\cE &\cE &\cL &\cH &    &\cE &\cE &\cE &\cE &\cE &\cL &\cH & \\
       &\ac &\ac &\ac &\ac &\ac &\ac &\ac &    &\cu&\ac &\ac &\ac &\ac &\ac &\ac &    &\\
   13  &\cE &\cE &\cE &\cE &\cE &\cL &\cE &\cH &    &\cE &\cE &\cE &\cE &\cL &\cE &\cH & \\
       &    &    &    &    &    &    &    &    &    &    &    &    &    &    &    &    & \nc \\
   14  &\cE &\cE &\cE &\cE &\cE &\cL &\cE &\cH &    &\cE &\cE &\cE &\cE &\cL &\cE &\cH & \\
       &    &    &\ac &\ac &\ac &\ac &\ac &    &    &    &    &    &    &\ac &\ac &    & \\
   15  &\cE &\cE &\cE &\cE &\cL &\cE &\cE &\cH &    &\cE &\cE &\cE &\cE &\cL &\cE &\cH & \\
       &\ac &\ac &    &    &    &    &    &    &    &\ac &\ac &\ac &\ac &    &    &    & \\
   16  &\cE &\cE &\cE &\cE &\cL &\cE &\cE &\cH &    &\cE &\cE &\cE &\cL &\cL &\cE &\cH & \\
       &    &    &\ac &\ac &\ac &\ac &\ac &    &    &    &    &    &    &\ac &\ac &    & \\
   17  &\cE &\cE &\cE &\cL &\cE &\cE &\cE &\cH &    &\cE &\cE &\cE &\cL &\cE &\cE &\cH & \\
       &    &    &    &    &    &    &    &    &    &    &    &    &    &    &    &    & \nc \\
   18  &\cE &\cE &\cE &\cL &\cE &\cE &\cE &\cH &    &\cE &\cE &\cE &\cL &\cE &\cE &\cH & \\
       &\ac &\ac &\ac &\ac &\ac &\ac &\ac &    &\cu &\ac &\ac &\ac &\ac &\ac &\ac &    & \\
   19  &\cE &\cE &\cL &\cE &\cE &\cE &\cE &\cH &    &\cE &\cE &\cL &\cE &\cE &\cE &\cH & \\
       &    &    &    &    &    &    &    &    &    &    &    &    &    &    &    &    & \nc \\
   20  &\cE &\cE &\cL &\cE &\cE &\cE &\cE &\cH &    &\cE &\cE &\cL &\cE &\cE &\cE &\cH & \\
       &    &    &\ac &\ac &\ac &\ac &\ac &    &    &    &    &    &    &\ac &\ac &    & \\
   21  &\cE &\cE &\cL &\cE &\cE &\cE &\cE &\cH &    &\cE &\cE &\cL &\cE &\cE &\cE &\cH & \\
       &\ac &\ac &    &    &    &    &    &    &    &\ac &\ac &\ac &\ac &    &    &    & \\
   22  &\cE &\cL &\cL &\cE &\cE &\cE &\cE &\cH &    &\cE &\cL &\cE &\cE &\cE &\cE &\cH & \\
       &    &    &\ac &\ac &\ac &\ac &\ac &    &    &    &    &    &    &\ac &\ac &    & \\
   23  &\cE &\cL &\cE &\cE &\cE &\cE &\cE &\cH &    &\cE &\cL &\cE &\cE &\cE &\cE &\cH & \\
       &    &    &    &    &    &    &    &    &    &    &    &    &    &    &    &    & \\
      
    \end{tabular}
    \caption{Reflection of a \enquote{fast} signal at the right border.
      We use $m_c=6$, $q=2$, hence $b_q=3$, $h=1$ and $h+1=2$.
      Therefore at the right end there are always $h+1=2$ cells with period
      $2$.
      We have introduced small spaces to make the boundaries of the last
      full block of cells with period $2$ better visible.
      On the left hand side the last full block has $3$ cells of period $2$
      and on the right hand side the last full block has $2$ cells of period $3$.
      For more details see the proof of Theorem~\ref{thm:lower-bound}.}
    \label{fig:reflection}
  \end{figure}

  
  Summing up all terms for the movement of a signal from the very
  first cell to the right border and back results in a time of
  \[
    t=1 +m\cdot m_c + m\cdot m_c + h\cdot q + q + h\cdot q +m\cdot m_c
    + m\cdot m_c+p_2+1
  \]
  This contains twice the term
  \[
    t' = m\cdot m_c + m\cdot m_c + h\cdot q= \sum_{i=3}^{n-1} p_i =
    \left(\sum_{i=1}^{n} p_i\right) -p_1 -p_2 -p_n
  \]
  hence
  \begin{align*}
    t = 1 + t' + p_n + t' +p_2+1
    & = 1 + 2\left(\sum_{i=1}^{n} p_i\right) -2p_1-2p_2-2p_n+p_n+p_2+1\\
    & = 2\left(\sum_{i=1}^{n} p_i\right) -2p_1-p_2-p_n+2\\
  \end{align*}
  as claimed.
  
  It can be observed that in the case that all $p_i=1$ the formula
  becomes the well-known lower bound of  $2n-2$.
\end{proof}

In the following section we will describe an algorithm which achieves
a running time which is slower than the lower bound in
Theorem~\ref{thm:lower-bound} by only a constant summand.

\section{Detailed description of the synchronization algorithm for
  \msca}
\label{s:details}

To the best of our knowledge the paper by Manzoni and Umeo
\cite{Manzoni_2014_FSSP_ar} is the only work on the FSSP in
one-dimensional \msca{} until now.
They describe an algorithm which achieves synchronization in time
$n\cdot \pmax$ where
$\pmax = \max \{\pi_{p}(c(i)) \mid 1\leq i\leq n\}$ is maximum
period used by some cell in the initial configuration $c$.

Below we will describe an algorithm which proofs the following:

\begin{theorem}
  \label{thm:upper-bound}
  For each $P$ there is a constant $d$ and an algorithm which
  synchronizes each \msfssp{} instance $c$ of size $n$ with periods
  $p_1, \dots, p_n \in P$ in time
  \begin{equation}
    \left( 2\sum_{i=1}^{n}  p_i \right) + d
    \label{eq:upper-bound}
  \end{equation}
\end{theorem}
In the case of standard CA all $p_i=1$ and
formula~(\ref{eq:upper-bound}) becomes $2n+d$ which is only a
constant number of $d+2$ steps slower than the fastest algorithms
possible.

\subsection{Core idea for synchronization}
\label{subs:core-idea}

In the proof of a lower bound above we have constructed problem
instances consisting of blocks consisting of cells with identical
period.
The arrangement was chosen in such a way that a signal, even if it
were to move as fast as possible, would have to spend $p$ steps in a
cell with period $p$ before moving on.
In a standard CA this is the time a signal with speed $1$ needs to
move across $p$ cells.
Which leads to the idea to have each cell with period $p$ of the
\msca{} simulate $p$ cells of a standard CA (solving the FSSP).
We'll call the simulated cells \emph{virtual cells} or \emph{v-cells}
for short, and where disambiguation seems important call the cells of
the \msca{} \emph{host cells}.
States of v-cells will be called v-states.

\subsection{Details of the synchronization algorithm}
\label{subs:details}

From now on, assume that we are given some standard CA for the
standard FSSP.
Its set of states will be denoted as $Q$.

\begin{algorithm}
  \label{alg:msfssp}
  As a first step we will sketch the components of the set $S$ of states
  of the \msca.

  We already mentioned in Subsection~\ref{subs:fssp-msca} that we assume
  $S$ to be of the form $S=\Per\x S_1$.
  Let $\tc=\lcm\Per$ denote the least common multiple of all
  $p\in \Per$.
  Since in the algorithm below host cells will have to count from $0$ up
  to $\tc-1$, we require that the set of states always contains a
  component $T=\{i\in \Z\mid 0\leq i< \tc\}$.
  Hence $S=\Per\x T \x S_2$, and we assume that the transition function
  will in each step update the $T$-component of a cell by incrementing
  it by its period, modulo $\tc$.
  Imagine that this is always \enquote{the last part} of a transition
  whenever a cell is active.
  Thus an active cell can identify the common updates by the fact that
  its $T$-component is $0$.
  But of course it is equally easy for an active cell to identify an
  activation that is \emph{the last before} a common update.

  Next, each host cell will have to store the states of some v-cells.
  As will be seen this will not only comprise the states of the $p$
  v-cells it is going to simulate, but also the states of v-cells
  simulated by neighboring host cells; we will call these
  \emph{neighboring v-cells}.
  To this end we choose
  $S=\Per\x T \x Q^{\leq\tc} \x Q^{\leq\pmax} \x Q^{\leq\tc} \x S_3$.
  We will denote the newly introduced components of a cell as $x$, $y$,
  and $z$.
  In $x$ a host cell will accumulate the states of more and more
  neighboring v-cells from the left.
  Analogously, in $z$ a host cell will accumulate the states of more and
  more neighboring v-cells from the right.
  In the middle component $y$ a host cell will always store the states
  of the $p$ v-cells it has to simulate itself.

  The simulation will run in cycles each of which is $\tc$ steps long
  and begins with a common update.
  During one cycle a cell with period $p$ will be active $\tc/p$ times.
  Whenever a host cell is active it collects as many neighboring
  v-states as possible, but at most $\tc$ from either side.
  More precisely this is done as depicted in the following table: \\[-1mm]

  \begin{center}
    \begin{tabular}{|C|C|C|}
      \hline
      (x_{-1},y_{-1},z_{-1}) & (x_{0},y_{0},z_{0})  &(x_{1},y_{1},z_{1}) \bigstrut  \\
      \hline
      \multicolumn{1}{c}{ } & \multicolumn{1}{c}{ }  &  \multicolumn{1}{c}{ }  \\[-2mm]
      \cline{2-2}
      \multicolumn{1}{c|}{ } & (\sfx_{\tc}(x_{-1}y_{-1}) ,y_{0}, \pfx_{\tc}(y_{1}z_1)  &  \multicolumn{1}{c}{ } \bigstrut \\
      \cline{2-2}
    \end{tabular}\\[1mm]
  \end{center}

  In other words, the states of the own v-cells are not changed, but
  more and more neighboring v-states are being collected.
  We will show in Lemma~\ref{lem:fast-collect} below that during the
  last activation of a cycle, \ie the last activation before a common
  update, after having collected neighboring v-states, the length of
  the $x$ and $z$ components are in fact $\tc$ and not shorter.
  It is therefore now possible for each host cell to replace the
  v-states of its $p$ v-cells by the v-states those v-cells would be in
  after $\tc$ steps.
  The $x$ and $z$ components are reset to the empty word.

  It is during the last activation of a cycle that a host will compute
  state $\sF$ for each of its v-cells.
  The immediately following activation is a common update for all host
  cells.
  They will simultaneously detect that their v-cells reached the
  \enquote{virtual \sF} and all enter the \enquote{real firing state}.
\end{algorithm}

For a proof of Theorem~\ref{thm:upper-bound} only the following two
aspects remain to be considered.

\begin{lemma}
  \label{lem:fast-collect}
  After one cycle of algorithm~\ref{alg:msfssp} each host cell will
  have collected the states of $\tc$ neighboring v-cells, to the left
  and to the right.
\end{lemma}

\begin{proof}
  Without loss of generality we only consider the case to the left.
  We will prove by induction on the global time that for all
  $\bar{t}\in\N_0$ the following holds:

  For $t=\bar{t}\bmod\tc$ for each cell with components $(x,y,z)$ as
  above and with period $p$ and for all $j\in\N_0$ with
  $0\leq j \leq t$: If $jp=t$ then the cell is active and
  after the transition $|x|\geq \min(jp, \tc)$.

  If $\bar{t}=0$ then $t=0$, $j=0$, and obviously $|x|\geq 0$ holds.

  Now assume that the statement is true for all times less or equal
  some $\bar{t}-1$.
  Again, nothing has to be done if $t=0$; assume therefore that $t>0$.
  
  Consider a cell with components $(x,y,z)$ and period $p$ and its
  left neighbor with components $(x',y',z')$ and period $q$, and
  therefore $|y'|=q$.
  Let $\bar{t}'$ be the time when the left neighbor was active for the
  last time before $\bar{t}$, and let $t'=\bar{t}'\bmod\tc$.
  Then $t'<t$, $t'=kq$ for some $k$, and since it was the
  \emph{last} activation before $t$, $t'+q=(k+1)q\geq t$.
  By induction hypothesis the left neighbor already had
  $|x'|\geq \min(kq,\tc)$.
  The new $x$ of the cell under consideration is $x=\sfx_{\tc}(x'y')$
  which then has length at least
  $\min(\tc,|x'y'|) \geq \min(\tc,\min(kq,\tc)+q) = \min(\tc,(k+1)q)$.
  Since $(k+1)q\geq jp$ the proof is almost complete.
  
  Strictly speaking the above argument does not hold when the left
  neighbor is a border cell.
  But in that case a cell can treated by its neighbor as if that has
  already $x$ filled with $\tc$ states $\sB$.
\end{proof}

\begin{lemma}
  For a problem instance $c$ of size $n$ with periods $p_1,\dots,p_n$
  the time needed by Algorithm~\ref{alg:msfssp} for synchronization
  can be bounded by 
  \begin{equation}
    \left( 2\sum_{i=1}^{n}  p_i \right) + d \;.
    \label{eq:upper-bound-2}
  \end{equation}
\end{lemma}

\begin{proof}
  The total number of v-cells simulated is $k=\sum_{i=1}^{n}  p_i$.
  A time optimal FSSP algorithm for standard CA needs $2k-2$ steps for
  the synchronization of that many cells.\
  During each cycle of length $\tc$ exactly $\tc$ steps of each v-cell
  are simulated, except possibly for the last cycle.
  During that, the $\sF$ v-state may be reached in less than $\tc$
  steps.

  Hence the total number of steps is
  $\tc \cdot \lfloor (2k-2) / \tc\rfloor +\tc +1 \leq 2k+d = \left(
    2\sum_{i=1}^{n} p_i \right) +d $ for some appropriately chosen
  constant $d$.
\end{proof}
To sum up taking together Theorem~\ref{thm:lower-bound} and
Theorem~\ref{thm:upper-bound} one obtains
\begin{corollary}
  For each $P$ there is a constant $d$ such that there is an \msca{}
  for the \msfssp{} which needs synchronization time
  $\left( 2\sum_{i=1}^{n} p_i \right) + d$ and for infinitely many
  sizes $n$ there are $a^n$ problem instances ($a>1$) for which there
  is a lower bound on the synchronization time of
  $\left( 2\sum_{i=1}^{n} p_i \right) - d$.
\end{corollary}

\section{Outlook}
\label{s:outlook}

In this paper we have described a \msca{} for the synchronization
problem which is sometimes faster and never slower than the one by
Manzoni and Umeo.
For a number of problem instances which is exponential in the number
of cells to be synchronized the time needed is close to some lower
bound derived in Section~\ref{s:lower-bound}.
An initial version of the proof could be improved thanks to an
anonymous reviewer.

While higher-dimensional \msca{} have been considered
\cite{Manzoni_2016_FSSP_ip}, in the present paper we have restricted
ourselves to the \emph{one-dimensional} case.
In fact it is not completely clear how to generalize the algorithm
described above to two-dimensional CA.
The \msca{} described in this paper it is essential that
\begin{itemize}
\item from one cell to another one there is only one shortest path
\item and it is clear how many v-cells a cell should simulate.
\end{itemize}
The generalization of this approach to $2$-dimensional CA is not
obvious for us.
In addition the derivation of reasonably good lower bounds on the
synchronization times seem to be more difficult, but if one succeeds
that might give a hint as to how to devise an algorithm.
As a matter of fact, the same happened in the one-dimensional setting.

Similarly it is not clear how to apply the ideas in the case of CA
solving some \emph{other} problem, not the FSSP, because only (?) for
the FSSP it is obvious which state(s) to choose for the v-cells in the
initial configuration.

Both aspects, algorithms and lower bounds, are interesting research
topics but need much more attention.
Even in the $1$-dimensional case there is still room for improvement
as has been seen in Figure~\ref{fig:signal-not-slow}.

It remains an open problem how cells with different update periods
should be ordered to ensure that they can be synchronized as soon as
possible.


\begin{thebibliography}{1}

\bibitem{Manzoni_2016_FSSP_ip}
Manzoni, L., Porreca, A.E., Umeo, H.: The firing squad synchronization problem
  on higher-dimensional {CA} with multiple updating cycles.
\newblock In: Fourth International Symposium on Computing and Networking,
  {CANDAR} 2016, Hiroshima, Japan, November 22-25, 2016, pp. 258--261 (2016)

\bibitem{Manzoni_2014_FSSP_ar}
Manzoni, L., Umeo, H.: The firing squad synchronization problem on {CA} with
  multiple updating cycles.
\newblock Theor. Comput. Sci. \textbf{559}, 108--117 (2014)

\bibitem{Umeo_2009_FSSP_ic}
Umeo, H.: Firing squad synchronization problem in cellular automata.
\newblock In: Encyclopedia of Complexity and Systems Science, pp. 3537--3574.
  Springer (2009)

\bibitem{Umeo_2020_HSCA_ar}
Umeo, H.: How to synchronize cellular automata -- recent developments --.
\newblock Fundam. Inform. \textbf{171}(1-4), 393--419 (2020)

\end{thebibliography}
\end{document}